\documentclass[conference]{IEEEtran}

\usepackage{graphicx}
\usepackage{float}
\usepackage{array}
\usepackage{tabularx}
\usepackage{times}
\usepackage{amsmath}
\usepackage{amssymb}
\usepackage{adjustbox}

\usepackage{verbatim}
\usepackage{moreverb}
\usepackage{alltt}
\usepackage{cite}
\usepackage{epsfig}
\usepackage{subfigure}
\usepackage{multirow}
\usepackage{indentfirst}
\usepackage{epstopdf}
\usepackage{amsfonts}
\usepackage[ruled,vlined]{algorithm2e}
\usepackage{url}

\newtheorem{theorem}{\textbf{Theorem}}
\newtheorem{lemma}{\textbf{Lemma}}
\newtheorem{corollary}{\textbf{Corollary}}
\newtheorem{definition}{\textbf{Definition}}
\allowdisplaybreaks[4]

\begin{document}

\title{Coded Caching with Heterogenous Cache Sizes}

\vspace{-0.25cm}


\author{\IEEEauthorblockN{Sinong Wang, Wenxin Li, Xiaohua Tian, Hui Liu}

\IEEEauthorblockA{\IEEEauthorrefmark{0}Department of Electronic Engineering, Shanghai Jiao Tong University\\
\{snwang,xtian,huiliu\}@sjtu.edu.cn}
}

\maketitle

\vspace{-0.25cm}
\begin{abstract}
We investigate the coded caching scheme under the heterogenous cache sizes.
\end{abstract}

\section{Introduction \label{sec1}}

Coded caching is a novel mechanism to relieve wireless congestion during peak-traffic times for content distribution~\cite{codedcaching1}\cite{codedcaching2}, where the temporal variability of wireless traffics is utilized. With coded caching, popular contents are partially prefetched at users' local cache during the placement phase, i.e., off-peak traffic times, and the rest of the contents are delivered using coded multicasting during delivery phase, i.e., peak traffic times upon request. Compared with the traditional caching scheme that adopts the orthogonal unicasting transmission and the caching gain is straightforwardly dependent on the user's cache size~\cite{caching1}-\cite{caching5}, coded caching provides coding opportunities among different requests during the delivery phase, which further exploits cache resources by jointly optimizing the placement and delivery phases. Efforts have been made to reveal the fundamental limits of coded caching in an information-theoretical perspective that the coded caching scheme in the bottleneck network can achieve a global cache gain to reduce the delivery-phase traffic volume and only exhibits the constant gap to the information-theoretical lower bound.

The seminal works attract much attentions in the community and encourages further investigations on the coded caching scheme~\cite{codedcachingnonuniform}-\cite{lowerbound}; however, existing works have the shared assumption, that is, users have the identical cache size, which is extremely difficult to satisfy in practice. In this paper, we present a comprehensive study on coded caching with heterogeneous cache sizes at the user end.

Several non-trivial challenges need to be addressed before we get the insight of coded caching scheme with heterogeneous cache size. First, what is the fundamental bound in an information theoretical perspective under heterogeneous cache size? Second, how to implement this lower bound by designing the placement and delivery schemes? In particular, could the schemes designed for the setting of homogeneous cache size be applied to the heterogeneous cache size scenario? If not, what is the root cause and how could we design algorithms dedicated for the heterogeneous cache size scenario? Third, how the heterogeneous cache size influences the wireless traffic volume during the delivery phase? Is it possible to obtain an analytical form trade-off between cache size and traffic volume as in the homogeneous setting~\cite{codedcachingnonuniform}-\cite{codedcachingmulti}?

This paper tries to shed light on how to resolve these challenges, where some interesting results have been derived. In the traditional coded caching scheme~\cite{codedcaching2}, it is assumed that each user has an identical memory space and the random caching procedure in the placement phase will produce the content segments of approximate equal size, and form a maximal clique of different segments that can be fully utilized to create the coding opportunities in the delivery phase. However, heterogeneous cache size incurs that the size of content segments in the delivery phase is also heterogeneous, which causes problems for coding. A straightforward solution is to perform padding to smaller-sized segments so that all segments can be aligned for coding. Apparently, such an approach can cause larger-sized segments to miss coding opportunities in the delivery phase thus increase the traffic volume.

Through analyzing the internal structure of such problem, we prove above zero-padding solution is the optimal coded delivery scheme under the random caching procedure, and such miss coding phenomenon is impossible to counteract in this setting. Then we drive the information-theoretical lower bound in this case, an interesting finding from our investigation is revealed that: although introducing potential miss of coding opportunities, coded caching scheme adopting padding under heterogeneous cache size still presents a constant gap to the optimal scheme, and the constant gap is less than 12, which is identical to the homogeneous case. The main reason is that the miss of coding opportunities is an inherent limitation of bottleneck network under the heterogenous cache sizes, which appears not only in the coded caching scheme, but also in the fundamental bound.

To further investigate the fundamental limits of heterogenous cache size, we introduce the concept of probabilistic cache set and characterize such memory-traffic volume tradeoff and order optimality via the numerical statistics of the user cache size distribution. We analytically show that the gap of traffic volume produced by coded caching scheme to the lower bound will decrease when the deviation of all users' cache sizes increases. This result implies that the fundamental bound has a more increasing speed  of traffic volume and the coded caching scheme will gradually degenerate to the uncoded version as the difference among users' cache sizes increases.

Besides that, we find that the coded caching scheme will present the characteristic of the grouping coded delivery (GCD), where users are divided into groups based on their cache sizes and coded caching are performed on each group separately. In particular, when the deviation of users's cache sizes increases, the GCD can be approximately order optimal to the lower bound. This finding indicates that coded caching scheme could be implemented in GCD under real systems, because GCD can gain significant decrease in computational complexity at the cost of very limited performance.

The remainder of the paper is organized as follows. We describe the service model and problem setting in Section~\ref{sec2} and provide some preliminaries and motivations in Section~\ref{sec3}. Section~\ref{sec4} presents our main results under heterogenous cache sizes, including our modified coded caching scheme, traffic volume-memory trade-off and order optimality analysis. Section~\ref{sec5} further investigates this problem under the probabilistic cache set. Numerical analysis are presented in Section~\ref{sec6}. Section~\ref{sec6} concludes this work and exhibits some interesting extensions. The proof of our main results and details are provided in the technical report.

\section{Model and Problem Statement \label{sec2}}

In this paper, we consider a set of users connecting to a content server through a shared wireless link that is similar as the setup in~\cite{codedcaching2}.

\subsection{Service Model}

We consider a network consisted of a content server connected to $K$ users through a shared, error-free link, as illustrated in Fig.~\ref{fig:sysmodel}. The error-free link can be achieved with error correction scheme or reliable transmission scheme in the upper layer. The user set is denoted by $\mathcal{K}=\{1,\ldots,K\}$. The content server has accessed to a database of $N (K\leq N)$ uniform distributed contents $W_{1}, W_{2}, \ldots, W_{N}$ with each of size $F$ bits.\footnote{The size can also be packet-based. For simplicity, we use the bits as the metric in the following.} The same-size-content assumption is for the theoretical convenience, which however does not hinder the practicability of the coded caching operations in the real world, because the main body of content objects can be tailored as the same size for coded caching based distribution and the rest in a small quantity can be distributed in the traditional way. The content index set is denoted by $\mathcal{N}=\{1,\ldots, N\}$. Each user $k$ has an isolated cache space $Z_{k}$ of size $M_{k}F$ bits for some real number $M_{k}\in [0,N]$. All users' cache sizes constitutes a cache set $\mathcal{M}=\{M_{1}, \ldots, M_{K}\}$\footnote{Here we use the cache set to denote the set consisted of users' cache sizes instead of each user's local cache $Z_{k}$.}. Without loss of generality, we assume that the cache set $\mathcal{M}$ is a ordered set, i.e., $M_{1}\leq M_{2} \leq \cdots M_{K}$. The system operates in two phases: a \emph{placement} phase and a \emph{delivery} phase.

\begin{figure}[htb]
\vspace{-1.5cm}
 \centerline{ \includegraphics[angle=0,width=0.70\textwidth]{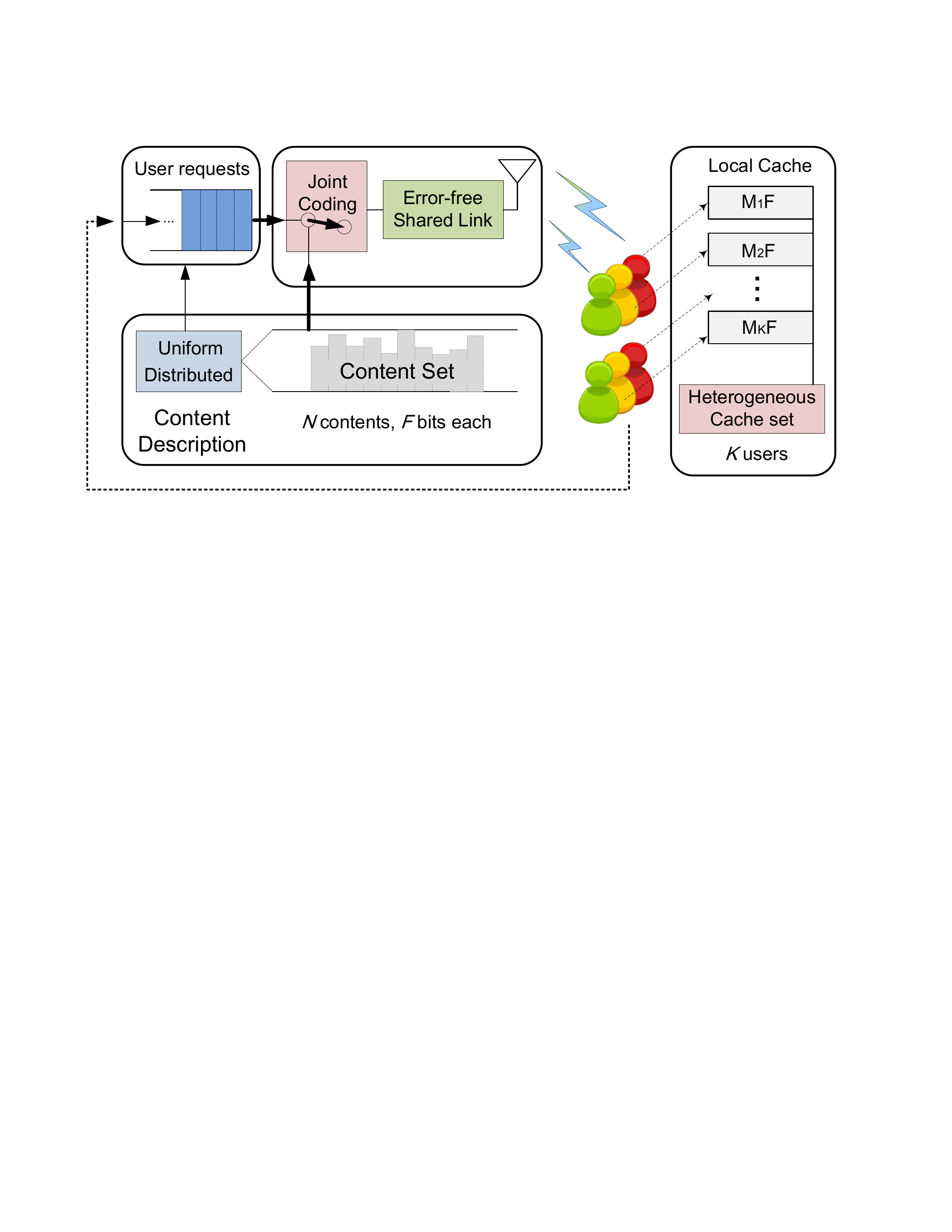}}
 \vspace{-13cm}
  \caption{Network architecture of coded caching.}
  \label{fig:sysmodel}
   \vspace{-0.4cm}
 \end{figure}

In the placement phase, the content server push contents $W_{1}, W_{2}, \ldots, W_{N}$ to the shared link and the caching of each user is done in the decentralized manner that the sever has no control over what parts of content goes into each user's local cache. The users divide their cache space into $N$ identical parts and randomly choose which bits to cache using a random number generator. By uploading the seed value of each user's random number generator, the server can reconstruct the caching contents of each user~\cite{codedcaching2}.

In the delivery phase, each user first sends its request $d_{k}$ ($d_{k}$ is the index of content $W_{d_{k}}$, $d_{k}\in \mathcal{N}, k\in \mathcal{K}$). These requests are independently and identically distributed (i.i.d) across the contents and users.  Then the server collects all users' requests $(d_{1}, d_{2}, \cdots, d_{K})$ and proceeds by transmitting a signal $X_{(d_{1}, d_{2}, \cdots, d_{K})}$ of size $R_{\mathfrak{F}}^{(d_{1}, d_{2}, \cdots, d_{K})}F$ bits over the shared link. This metric are referred to as the load or the traffic volume of the shared link under scheme $\mathfrak{F}$. Using the caching contents and signal received over the shared link, each user can reconstruct its requested contents.

\subsection{Problem Statement}

Then based on above setting, we present the basic definitions in our problem.

A memory-traffic volume pair $(\mathcal{M},R_{\mathfrak{F}}^{(d_{1}, \cdots, d_{K})})$ is \emph{achievable for scheme $\mathfrak{F}$ under requests $(d_{1}, \cdots, d_{K})$} if every user $k$ is able to reconstruct its requested content $W_{d_{k}}$ with error probability $P_{e}\rightarrow 0$ and produce the traffic volume $R_{\mathfrak{F}}{(d_{1}, \cdots, d_{K})}F$ bits under the cache set $\mathcal{M}$. A memory-traffic volume pair $(\mathcal{M},R^{\mathfrak{F}})$ is \emph{achievable for scheme $\mathfrak{F}$} if this pair is achievable for every possible requests $(d_{1}, \cdots, d_{K})$ in the delivery phase. Defined by
\begin{definition}(Achievable scheme)
\begin{equation}
R_{\mathfrak{F}}\triangleq \max\limits_{(d_{1}, \cdots, d_{K})\in \mathcal{N}^K}R_{\mathfrak{F}}^{(d_{1}, \cdots, d_{K})}
\end{equation}
the worst case normalized traffic volume for scheme $\mathfrak{F}$.
\end{definition}

We use the $R^{*}$ to represent the smallest traffic volume such that $(\mathcal{M}, R^{*})$ is \emph{achievable}. Defined by
\begin{definition}(Optimal scheme)
\begin{equation}
\left(\mathcal{M}, R^{*}\right)\triangleq \inf \left\{\left(\mathcal{M},R_{\mathfrak{F}}\right),\forall \mathcal{M}, \mathfrak{F} \right\}
\end{equation}
the infimum of all achievable $\left(\mathcal{M}, R_{\mathfrak{F}}\right)$.
\end{definition}

Clearly, $R_{\mathfrak{F}}$ is function of cache set $\mathcal{M}$ and number of users $K$ and number of contents $N$. To emphasize this dependency, we rewrite above traffic volume as $R_{\mathfrak{F}}(\mathcal{M},N,K)$ and $R^{*}(\mathcal{M},N,K)$. The aim of this paper is to find a scheme $\mathfrak{F}$ such that $R_{\mathfrak{F}}(\mathcal{M},N,K)$ guarantees the order optimality, defined as
\begin{definition}(Order optimality)
The scheme $\mathfrak{F}$ is order optimal if only if
\begin{equation}\label{eq:deforderopt}
\frac{R_{\mathfrak{F}}(\mathcal{M},N,K)}{ R^{*}(\mathcal{M},N,K)} \mathop  \leq  C,
\end{equation}
$C$ is a constant independent of the system parameters $\mathcal{M}, N$ and $K$.
\end{definition}

We can see that the order optimality can be guaranteed only if the traffic volume produced by scheme $\mathfrak{F}$ has the constant gap to the optimum. For the simplicity of the following illustration, we refer $\mathcal{M}^{I}$ as the homogeneous cache set for the heterogenous cache set $\mathcal{M}$, both of which have the identical aggregate cache size. Defined by
\begin{definition}(Homogeneous cache set)
The homogeneous cache set of $\mathcal{M}$ is defined as $\mathcal{M}^{I}=\{\overline{M},\overline{M},\ldots,\overline{M}\}$, where $\overline{M}=\mathbb{E}_{k}[M_{k}]$ the average size of all users' local caches.
\end{definition}

\section{Preliminary and Motivation \label{sec3}}

\subsection{Related Works}

The wireless traffics presents a high temporal variability in network traffic volume and the caching is a promising way to balance such traffic load. One line of studies~\cite{caching1}-\cite{caching5} were centered on optimizing the system throughput of the caching network without considering the coding in the transmission, where the performance is limited by the size of each user. The second line of researchers, recently, investigate such problem in the information-theoretical perspective and propose a novel technique named as coded cache. This technique derives from the index coding problem, which is to determine the minimum code length to satisfy multiple given users requests under given side information, i.e., cache sate, in the broadcasting channel~\cite{indexcoding1}-\cite{indexcoding3}. The coded caching scheme is a two-phases approach including the side information choice (placement phase) and index coding delivery (delivery phase). Through jointly optimizing these two phases, the code length under any user requests can be minimized and the system performance, especially in the aspect of traffic volume, can be improved significantly compared to the uncoded counterpart.

Due to the significant gain in reducing traffic volume and the straightforward analytical form, this result has attracted much attention in the community. There are three sub-lines of this feild. The first line of work~\cite{lowerbound} mainly investigates more tight information-theoretical lower bound. The second line of studies are aim to apply the coded caching technique to other network structures such as hierarchical caching system in~\cite{codedcachinghier} and heterogenous network with multi-level cache access in~\cite{codedcachingmulti}. The third line of works focuses on the more practical scenarios. In~\cite{codedcachingnonuniform}--\cite{NonuniformJi}, the authors assume the content has the non-uniformly populated contents and propose the specific cache allocation strategy to minimize the delivery-phase traffic volume. In~\cite{codedcachingonline}, the authors consider the dynamic content set and design the cache updating scheme, termed as coded LRS, to approximate the lower bound. Following this line of research, in this paper, we consider the more general scenarios with heterogeneous cache sizes. We now briefly review the original coded caching scheme with an example.

\textbf{Example 1} (Decentralized Coded Caching Scheme $\mathfrak{F}$) Suppose the system is distributing $2$ contents A and B to $2$ users, each with the cache size $MF$ bits. The size of each content is also $F$ bits. In the placement phase, each user randomly caches $MF/2$ bits of content A and B independently. Let us focus on content $A$. The operations of placement phase partition content A into four subcontents, $A=(A_{{\O}},A_{1},A_{2},A_{1,2})$, where $U\subset\{1,2\}$, and $A_{U}$ denotes the segments of content $A$ that are prefetched in the memories of users in $U$. For example, $A_{1}$ represents the segments of $A$ only available in the memory of user 1. We use $|\cdot|$ to denote the expectation size of each segment, and $|A_{{\O}}|=(1-M/2)^2F\text{ bits}, |A_{1}|=|A_{2}|=(M/2)(1-M/2)F\text{ bits}$ and $|A_{1,2}|=(M/2)^2F\text{ bits}$. The same analysis holds for content B. In the delivery phase, the worst case is that user 1 and user 2 requesting content A and B, respectively. User 1 has cached subcontent $A_{1}$ and $A_{1,2}$ in the placement phase and lacks $A_{{\O}}$ and $A_{2}$. Similarly, user 2 has already cached $B_{2}$ and $B_{1,2}$ , and lacks $B_{{\O}}$ and $B_{1}$. With traditional uncoded caching scheme, the server is required to unicast $A_{{\O}}$ and $A_{2}$ to user 1 and unicast $B_{{\O}}$ and $B_{1}$ to user 2. The total traffic volume is
\begin{equation*}
2\left(\frac{M}{2}\right)\left(1-\frac{M}{2}\right)F+2\left(1-\frac{M}{2}\right)^2F=(2-M)F .
\end{equation*}
With the coded caching scheme, the server can satisfy the requests by transmitting $A_{{\O}}$, $B_{{\O}}$ and $A_{2}\oplus B_{1}$ over the shared link, where $\oplus$ denotes the bit-wise XOR operation. The traffic volume over the shared link is
\begin{align*}
\left(\frac{M}{2}\right)\left(1-\frac{M}{2}\right)F+2\left(\frac{M}{2}\right)^2F&= \left(1-\frac{M}{4}\right)(2-M)F.
\end{align*}

We can see that, the coded caching scheme has a coded gain of $\left(1-\frac{M}{4}\right)$ in contrast to the uncoded caching scheme. The traffic volume for general case can be seen in~\cite{codedcaching2}.

\subsection{Motivation}

In the regime of heterogenous cache set, an intuition is whether the decentralized coded caching scheme $\mathfrak{F}$ can be applied straightforward. In fact, making the following simple modifications of scheme $\mathfrak{F}$, we can adopt it in our regime.

\emph{In the placement phase, user $k$ randomly prefetches $M_{k}F/N$ bits of content $n$.}

Based on this caching strategy, the same content stored in different users' local caches will occupy the different size of memory space. Thus, in the delivery phase, for each multicast user group, the size of segment, i.e., $A_{2}$ and $B_{1}$,  for each user will be different. Since these segments should participate bit-wise XOR, all the segments in one group should be bits-padded to the longest one. A possible padding method is to pad zero bits to the shorter segments and we refer this scheme as $\mathfrak{F}_{o}$. The following example will illustrate the performance of scheme $\mathfrak{F}_{o}$.

\textbf{Example 2} (Coded Caching with Heterogenous Cache Set) Suppose that there are three similar system distributing $N$ contents to $K=2$ users. The first system adopts our modified coded caching scheme $\mathfrak{F}_{o}$ and the users have the cache set $\mathcal{M}=\{(2-\alpha)M,(2+\alpha)M\}F$ bits with $0<\alpha<2$. The second system divides users into two groups and adopts scheme $\mathfrak{F}_{o}$ in each group. The users take the same cache set $\mathcal{M}$. The third system adopts traditional coded caching scheme $\mathfrak{F}$, and have the uniform cache size $\mathcal{M}^{I}=\{2M,2M\}F$ bits. Remark that above three systems have the same aggregate cache size of $4MF$ bits.

Assume that user 1 and user 2 request content A and content B. Based on the coded delivery technique, the server transmitted signal $A_{2}\oplus B_{1}$, $A_{\emptyset}$ and $B_{\emptyset}$, and the corresponding signal size via first system is
\begin{equation}\label{eq:exmpletraffic1}
R_{\mathfrak{F_{o}}}(\mathcal{M},N,2)=\left(1-\frac{(2-\alpha)M}{N}\right)\left(2-\frac{(2+\alpha)M}{N}\right)
\end{equation}

The corresponding signal size via second system is
\begin{align}
R_{\mathfrak{U}}\left(\mathcal{M},N,2\right)&=\left(2-\frac{4M}{N}\right)\label{eq:exmpletraffic2}.
\end{align}

The signal size via third system is
\begin{align}
R_{\mathfrak{F}}\left(\mathcal{M}^{I},N,2\right)&=\left(1-\frac{2M}{N}\right)\left(2-\frac{2M}{N}\right)\label{eq:exmpletraffic3}.
\end{align}

Then we consider two kinds of scenarios: the cache set of small deviation ($\alpha=0.2$) and the cache set of large deviation ($\alpha=1.8$) to refer the situation that the size of each user's local cache is similar and extremely different, separately. The number of contents $N=4M$.

Base on the traffic volume formula (\ref{eq:exmpletraffic1})-(\ref{eq:exmpletraffic3}), we get the following results, shown in TABLE~\ref{tab:example}. It can be seen that, when the size of each user's local cache is similar, the traffic volume produced by system 1 approximates to the system 3 and much lower than system 2; when the size of each user's local cache is large, it approximates to the system 2 and much larger than the system 3. This result comes from a ``bits waste'' phenomenon that when the user cache size varies widely, the size of requested segments $A_{2}, B_{1}$ are largely different, and the scheme $\mathfrak{F}_{o}$ will pad lots of useless zero bits to the smaller segment $B_{1}$, which will diminish part of coding opportunities for the longer segment $A_{2}$ and thus increase the traffic volume.

\begin{table}[htb]
\vspace{-0.3cm}
\caption{The coded data and traffic volume incurred} \vspace{-0.5cm}
\begin{center}
\begin{tabular}{|c|c|c|c|}
\hline
Scenario/Traffic(bits) & system 1 & system 2 & system 3\\
\hline \hline
$\alpha=0.2$ & 0.7975F & 1F & 0.75F \\
\hline
$\alpha=1.8$ & 0.9975F & 1F & 0.75F \\
\hline
\end{tabular}
\end{center}
 \label{tab:example}
 \vspace{-0.5cm}
 \end{table}

Example 2 shows that when the cache set is approximately uniform, the coded caching scheme $\mathfrak{F}_{o}$ still has a significant caching gain, while when the cache set is extremely different, such gain will diminish and the grouping of users (system 2) will not reduce the performance of whole system. This phenomenon and the reason behind it inspires us in three dimensions. First, whether the scheme $\mathfrak{F}_{o}$ is still order optimal or whether can we develop some new techniques to overcome the miss of coding opportunities when the users' cache sizes are extremely different. Second, in the derivation of traffic formula (\ref{eq:exmpletraffic1}), there exists amount of maximization operation in traffic volume-memory trade-off and the number of such operation will increase exponentially with number of users $K$ increasing. The key is whether we can get a close-form expression of such trade-off. Third, what condition the cache set should satisfy such that coded caching scheme will have a similar performance of the GCD scheme. The following work answers these questions.

\section{Fundamental Bound and Implementation Scheme\label{sec4}}

In this section, we will derive the fundamental bound, design the corresponding scheme under the heterogenous cache sizes and prove its order optimality.

\subsection{Information-theoretical lower bound}

The information-theoretical lower bound is independent of any specific schemes, instead, only dependent on the system parameters including $\mathcal{M}, N$ and $K$. The following theorem gives this lower bound on the optimal achievable traffic volume $R^{*}(\mathcal{M},N,K)$ based on the Fano's inequality and cut set-bound argument~\cite{tcover}.

\begin{theorem}(Cut-set Bound)\label{thm:cutsetbound}
For caching problem with $N$ contents, $K$ users, and ordered cache set $\mathcal{M}=\{M_1,M_2,\ldots, M_K\}$, we have
\begin{align}
R^{*}(\mathcal{M},N,K)&\geq R^{c}(\mathcal{M},N,K)\nonumber\\
\triangleq & \max\limits_{s\in \mathcal{K}}\left\{\max\limits_{U\subseteq[K],|U|=s}\left\{s-\frac{\sum_{i\in U}M_{i}}{\lfloor N/s\rfloor}\right\}\right\}\label{eq:cutsetbound}.
\end{align}
\end{theorem}
\begin{proof}
We provide a proof sketch of Theorem~\ref{thm:cutsetbound}. We regard each user's cache and the content server as the nodes and the broadcasting links between the users and content server as the edges in the bipartite graph. Then, for a feasible $(\mathcal{M}, R)$ pair, the total information contained in the memory of any subset of caches and the server transmitted signals must be at least the size of contents that users accessing these caches can reconstruct. Using above concept and traverse all subsets of users, Theorem~\ref{thm:cutsetbound} follows. See~\cite{techreport} for the details.
\end{proof}
\textbf{Insight on the lower bound:}\\
\textbf{(1)} In the traditional homogeneous case, the lower bound achieves the approximately maximum value of $\frac{N}{4\overline{M}}$ when $s\approx \frac{N}{2\overline{M}}$, which implies that the optimal scheme attains the delivery-phase traffic volume that is inverse proportional to the cache size. The traditional coded caching scheme also attains such inverse proportional trade-off and only exhibits constant gap less than 12. While in the regime of heterogenous cache set, above analysis will become more complicated. Consider the assumption that $M_{1} \leq M_{2} \leq \cdots \leq M_{K}$, the inner maximization operation in the lower bound is reduced to $s-\sum_{i=1}^{s}M_{i}/\lfloor N/s\rfloor$. Then the cut set bound $R^{c}(\mathcal{M},N,K)$ can be regarded as the function of $s$ and it can be proved that $R^{c}(\mathcal{M},N,K)$ achieves the maximum when $s$ satisfies $\sum_{i=1}^{s}M_{i}+(s-1)M_{s}\leq N$ and $\sum_{i=1}^{s+1}M_{i}+sM_{s+1}> N$, which has a discrete form and is impossible to derive such a simple form as in the homogenous case. But we can use this condition to simplify such lower bound and make the order analysis in the sequel.\\
\textbf{(2)} As can be seen in (\ref{eq:cutsetbound}), when the users has the identical cache sizes, above lower bound will embody the bound derived in~\cite{codedcaching1} and we have the following result.
\begin{corollary}(Relation between two kinds of cut-set bound)\label{coro:twolowerbound}
For caching problem with $N$ contents, $K$ users, we have
\begin{equation}
R^{c}(\mathcal{M},N,K)\geq R^{c}(\mathcal{M}^{I},N,K),
\end{equation}
equal if only if $\mathcal{M}=\mathcal{M}^{I}$.
\end{corollary}

Corollary~\ref{coro:twolowerbound} shows that when the aggregate cache size is fixed, if we regard the lower bound as the function of the cache set $\mathcal{M}$, this function achieves the conditional minimum when cache set is homogenous. Namely, the cut-set bound will increase when the aggregate cache sizes is distributed in a nonuniform manner. This result can be regarded as a preliminary illustration for the degenerated performance under heterogenous cache set. The reason will be discussed in the last subsection and the further quantitative discussion will be presented in the Section~\ref{sec5}.

Previous work~\cite{codedcaching1}~\cite{lowerbound} points out the cut-set bound is sometimes loose, and tighter bounds on $R^{*}(\mathcal{M},N,K)$ can be derived via stronger arguments than the cut-set bound. However, the cut-set bound along is sufficient for illustration of the main idea of our work and we will use this bound in the following discussion.

\subsection{Coded Caching Scheme}

For clarity, we first present the main procedure of our modified scheme $\mathfrak{F}_{o}$ mentioned before. The pseudocode is shown in Algorithm 1. Then we utilize the inherent characteristic of scheme $\mathfrak{F}_{o}$ to prove the zero-padding solution is the optimal coded delivery scheme under the random caching procedure.

\begin{algorithm}[htb]
\caption{Decentralized coded caching scheme with nonuniform cache size $\mathcal{M}$.}
\textbf{Placement Phase}\\
\For{$(k=0; k<K; k++)$}{
\For{$(n=0; n<N; n++)$}{
user $k$ randomly prefetches $M_{k}F/N$ bits of content $n$ \;
}
}
\textbf{Delivery Phase}\\
\For{$(k=K, k>0; k--)$}{
\For{choose $k$ users from $K$ users to form a subset $U$}{
$Maxsize\leftarrow \max_{k \in U}|V_{k,U/\{k\}}|$\;
\For{$l \in U$}{
\If{$|V_{k,U/\{k\}}|<Maxsize$}{
$temp \leftarrow \left(Maxsize-|V_{k,U/\{k\}}|\right)$ bits of all zero.\;
$V_{k,U/\{k\}}\leftarrow V_{k,U/\{k\}}+temp$\;
}
$X_U\leftarrow X_U \oplus V_{k,U/\{k\}}$ \;
}
Multicast the coded data $X_U$ to users in $U$.
}
}
\end{algorithm}

Operator $\bigoplus$ refers to bitwise XOR operation. In the placement phase, the random caching procedure will divide content $i$ into $2^K$ segments: $W_{i,U},U\subset 2^{\mathcal{K}}$. In the delivery phase, the element $V_{k,U/\{k\}}$ in signal $X_U$ represents the segment of user $k$'s requested content that being cached in users of set $U/\{k\}$\footnote{If user k's request is $d_{k}$, then $V_{k,U/\{k\}}=W_{d_{k},U/\{k\}}$.}. Since we take the zero-bits-padding for each segment before delivery, the size of signal $X_{U}$ is determined by the longest segment $V_{k,U/\{k\}}, k \in U$.

\textbf{Example 3} ($K=3, \mathcal{M}=\{M_{1}, M_{2}, M_{3}\}$) Consider the caching problem with $N$ contents and $K=3$ users with cache set $\{M_{1}, M_{2}, M_{3}\}$. Based on the Algorithm 1, each content $k$ will be divided into 8 sub-contents such as $A=\{A_{{\O}}, A_{1},A_{2},A_{3},A_{12},A_{13},A_{23},A_{123}\}$, and the transmission signal and corresponding signal size are,
\begin{itemize}
  \item $|V_{1,23}\oplus V_{2,13}\oplus V_{3,12}|\mathop  = \limits^{(a)}|V_{1,23}|=\left(1-\frac{M_{1}}{N}\right)\left(\frac{M_{2}}{N}\right)$ $\left(\frac{M_{3}}{N}\right)$;
  \item $|V_{1,2}\oplus V_{2,1}|\mathop  = \limits^{(a)}|V_{1,2}|=\left(1-\frac{M_{1}}{N}\right)\left(\frac{M_{2}}{N}\right)\left(1-\frac{M_{3}}{N}\right)$;
  \item $|V_{3,1}\oplus V_{1,3}|\mathop  = \limits^{(a)}|V_{3,1}|=\left(1-\frac{M_{1}}{N}\right)\left(1-\frac{M_{2}}{N}\right)\left(\frac{M_{3}}{N}\right)$;
  \item $|V_{2,3}\oplus V_{3,2}|\mathop  = \limits^{(a)}|W_{2,3}|=\left(1-\frac{M_{1}}{N}\right)\left(1-\frac{M_{2}}{N}\right)\left(\frac{M_{3}}{N}\right)$;
  \item $|V_{1,{\O}},V_{2,{\O}},V_{3,{\O}}|=3\left(1-\frac{M_{1}}{N}\right)\left(1-\frac{M_{2}}{N}\right)\left(1-\frac{M_{3}}{N}\right)$;
\end{itemize}
The operation (a) is based on the zero-bit-padding technique. Assume the users request content A, B and C in the delivery phase. The segment $V_{1,23}=A_{23}$, $V_{2,13}=B_{13}$, and other segment has the same meaning. The following lemma shows an important characteristic of this longest segment.

\begin{lemma}(Ordered segment size)\label{lem:orderedsize}
Based on Algorithm 1, the size of each transmission signal $X_{U}$ is,
\begin{equation}
|X_{U}|=\max\limits_{k \in U}|V_{k,U/\{k\}}|=|V_{\inf U,U/\{\inf U\}}|.
\end{equation}
\end{lemma}

Lemma~\ref{lem:orderedsize} shows that the longest segment of signal $X_{U}$ is user $\inf U$'s requested content $V_{\inf U,U/\inf U}$, who has the minimum cache size of user set $U$. This result is intuitive, since the segment $V_{k,U/k}$ represents the logical unit of content $W_{d_{k}}$ that only stored in the caches of users in $U/k$ and related to the the cache size of these users, if user $k$ has the minimum cache size of all users in $U$, the users in $U/k$ will has the largest memory space to store this logical unit and yield the largest size of it. The strict proof can be seen in~\cite{techreport}.

Based on Lemma~\ref{lem:orderedsize}, we have the following theorem,
\begin{theorem}(Optimality of the coded delivery scheme)\label{thm:optcode}
	The zero-padding scheme in delivery phase is the optimal scheme under the placement phase scheme of Algorithm 1.
\end{theorem}
\begin{proof}
	Here is a proof sketch of Theorem~\ref{thm:optcode}, which is based on the concept of the index coding~\cite{indexcoding1}. As shown before, the placement phase scheme of Algorithm 1 will produce $N\cdot 2^K$ segments, and during the delivery phase, the users will request different contents in the worst case. We regard each segment as the node and there exists an arc from node $i$ to node $j$ if only if content segment $i$ has been stored in the user that requests content $j$. Besides, there also exists an arc between node $V_{k,U}$ and $V_{k,U'}$ if only if $U\subset U'$ or $U' \subset U$, because these two kinds of segment can be regarded in the same node in some cliques. Then, the optimal scheme is to find the maximal set cover\footnote{Here the set cover and cliques mentioned below are the concepts of the index coding and have the same meaning in the graph theory, further details can be seen in~\cite{indexcoding1}} of such graph according to the index coding theory, based on which we design the optimal coding scheme. Finally we prove that the delivery scheme in Algorithm 1 will incur identical traffic volume compared to the optimal coding scheme, Theorem~\ref{thm:optcode} follows. See~\cite{techreport} for the details.
\end{proof}

Theorem~\ref{thm:optcode} shows the optimality of coding scheme of Algorithm 1 with respect to the given caching procedure. Here is an illustrative example.

\begin{figure}[htb]
\vspace{-.8cm}
 \centerline{ \includegraphics[angle=0,width=0.60\textwidth]{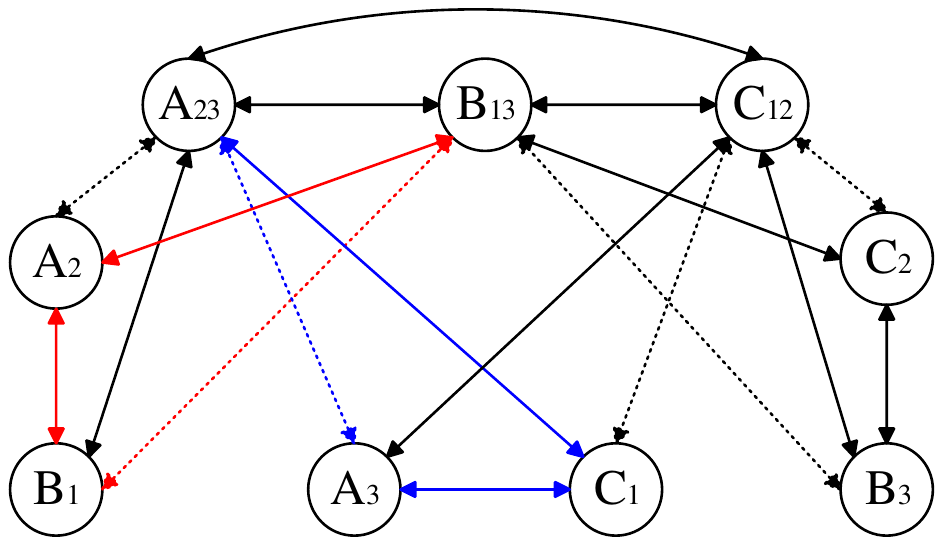}}
 \vspace{-12.2cm}
  \caption{Conflict graph for $K=3$ and $N=3$.}
  \label{fig:optcode}
   \vspace{-0.3cm}
 \end{figure}

\textbf{Example 4} (Optimality of the coded delivery scheme) Suppose the system is distributing $N=3$ contents A, B, and C to $K=3$ users with $M_{1}\leq M_{2} \leq M_{3}$. In the worst-case, the users will request content A, B and C. Then we construct the graph based on above proof procedure as illustrated in the Fig. Under the coded delivery scheme in Algorithm 1, it only encodes the following four cliques: $(A_{23}, B_{13}, C_{12})$, $(A_{1}, B_{2})$, $(A_{3}, C_{1})$ and $(B_{3}, C_{2})$. However, there exists other larger cliques such as $(B_{13}, A_{2}, B_{1})$ and $(A_{23}, A_{3}, C_{1})$ as shown in the red and blue line in the Fig~\ref{fig:optcode}. Here we draw a dotted line between node $B_{13}$ and $B_{1}$, since they can be regarded as the same node in the clique $(B_{13}, A_{2}, B_{1})$. We can find that larger cliques always occur among the different type segments, i.e., $A_{ij}$ and $A_{i}$. Thus, we can design the following optimal coding scheme that always pads the useful bits subtracting from $A_{ij}$ to $A_{i}$. Consider the bit-wise operation between $A_{2}$ and $B_{1}$, since the length of $B_{1}$ is shorter than $A_{2}$, the optimal coding scheme will pad the useful bits subtracting from segment $B_{13}$ to $B_{1}$, then the length of $B_{13}$ will be reduced. Based on this procedure, we have the results shown in TABLE~\ref{tab:example3}.

\begin{table}[htb]
\vspace{-0.3cm}
\caption{The coded data of scheme $\mathfrak{F}_{o}$ and optimal scheme} \vspace{-0.5cm}
\begin{center}
\begin{tabular}{|c|c|c|}
\hline
Index & Scheme $\mathfrak{F}_{o}$ & Optimal system\\
\hline \hline
1 & $A_{2} \oplus B_{1}$ & $A_{2} \oplus \{B_{1}'=B_{1}+B_{13}'\}$ \\
\hline
2 & $A_{3} \oplus C_{1}$ & $A_{3} \oplus \{C_{1}'=C_{1}+C_{12}'\}$ \\
\hline
3 & $B_{3} \oplus C_{2}$ & $B_{3} \oplus \{C_{2}'=C_{2}+C_{13}'\}$ \\
\hline
4 & $A_{23} \oplus B_{13} \oplus C_{12}$ & $A_{23}\oplus \{B_{13}-B_{13}'\}\oplus \{C_{12}-C_{12}'\}$\\
\hline
5 & $A_{{\O}},  B_{{\O}}, C_{{\O}}$ & $A_{{\O}}, B_{{\O}}, C_{{\O}}$\\
\hline
\end{tabular}
\end{center}
 \label{tab:example3}
 \vspace{-0.5cm}
 \end{table}

It can be seen that, in the 4th transmission, the length of $B_{13}$ and $C_{12}$ have been reduced due to the useful padding of previous transmission. However, the traffic volume produced in 4th transmission is not reduced because the length of the longest segment $A_{23}$ is not reduced. Thus, the scheme $\mathfrak{F}_{o}$ produces the identical traffic volume of the optimal scheme. The strict proof and details can be seen in.

\begin{lemma}(Traffic-memory trade-off)\label{lem:trafficformula1}
For caching problem with $N$ contents, $K$ users each with ordered cache set $\mathcal{M}=\{M_1,M_2,\ldots, M_K\}$, the scheme $\mathfrak{F}_{o}$ produced the traffic volume of
\begin{align}
R_{\mathfrak{F}_{o}}(\mathcal{M},N,K)&=\sum\limits_{s=1}^{K}\sum\limits_{i=1}^{s}\left[ \left(\prod\limits_{j=i+1}^{K}\frac{M_{j}}{N}\right)^{\dag}\cdot\prod\limits_{j=1}^{i}\left(1-\frac{M_{j}}{N}\right)\cdot\right.\notag\\
&\left. \sum\limits_{i+1\leq j_{1}\leq\cdots \leq j_{s-i} \leq K} \prod_{k=1}^{s-i}\left(\frac{N-M_{j_k}}{M_{j_k}}\right) \right]
\end{align}\label{eq:closetraffic}
where the $(x)^{\dag}=1$ if $x=0$; else $(x)^{\dag}=x$.
\end{lemma}

Lemma~\ref{lem:trafficformula1} presents a possible formula to calculate the traffic volume under ordered cache set $\mathcal{M}$, which is based on the natural summarization of the delivery scheme in Algorithm 1. The summarization procedure is tedious and details can be seen in~\cite{techreport}. However, the trade-off formula in Lemma~\ref{lem:trafficformula1} is difficult to calculate due to the exponential number of terms in the summation. And we cannot get any useful insight, i.e., how the user cache sizes influence the traffic volume, from such trade-off. Fortunately, this formula shows an useful iterative structure in aspect of number of users $K$ that we can utilize to get a close-form expression of the traffic volume under scheme $\mathfrak{F}_{o}$ and show some basic insights.

\begin{theorem}(Close-form expression of the traffic-memory trade-off)\label{thm:trafficformula2}
For caching problem with $N$ contents, $K$ users, with ordered cache set $\mathcal{M}=\{M_1,M_2,\ldots, M_K\}$, the scheme $\mathfrak{F}_{o}$ produces the traffic volume of,
\begin{equation}\label{eq:trafficformula2}
R_{\mathfrak{F}_{o}}(\mathcal{M},N,K)=\sum\limits_{i=1}^{K}\left[\prod\limits_{j=1}^{i}\left(1-\frac{M_{j}}{N}\right)\right].
\end{equation}
\end{theorem}
\begin{proof}
Here is a proof sketch of Theorem~\ref{thm:trafficformula2}. We use the mathematical induction method to simplify and induction basis is the number of users. First, the traffic volume under $k=1$ users is $1-M_1/N$; then we assume the traffic volume under $k=K-1$ users is $\sum_{i=1}^{K-1}\left[\prod_{j=1}^{i}\left(1-\frac{M_{j}}{N}\right)\right]$ and prove the traffic volume has the form of $(\ref{eq:trafficformula2})$ when $k=K$ based on the trade-off (\ref{eq:closetraffic}) and assumption when $k=K-1$, Theorem~\ref{thm:trafficformula2} follows. See~\cite{techreport} for the details.
\end{proof}

The proof of Theorem~\ref{thm:trafficformula2} and the traffic volume formula (\ref{eq:trafficformula2}) provide us with an observation that, in the aspects of $s$th transmission, the traffic volume due to the introduction of user $K$ can be seen as the average combination of $(s-1)$th and $s$th transmissions when there are $(K-1)$ users. The main reason is that the random caching in the placement phase and coded multicasting procedure in the delivery phase ``connects'' the users' local caches together, such that the requested sent from a new user can be not only satisfied by its own cache but also by other users' caches.

\subsection{Discussion}

We use the close-form trade-off to derive the following equivalent network model and scheme.

\begin{figure}[htb]
\vspace{-0.6cm}
 \centerline{ \includegraphics[angle=0,width=0.5\textwidth]{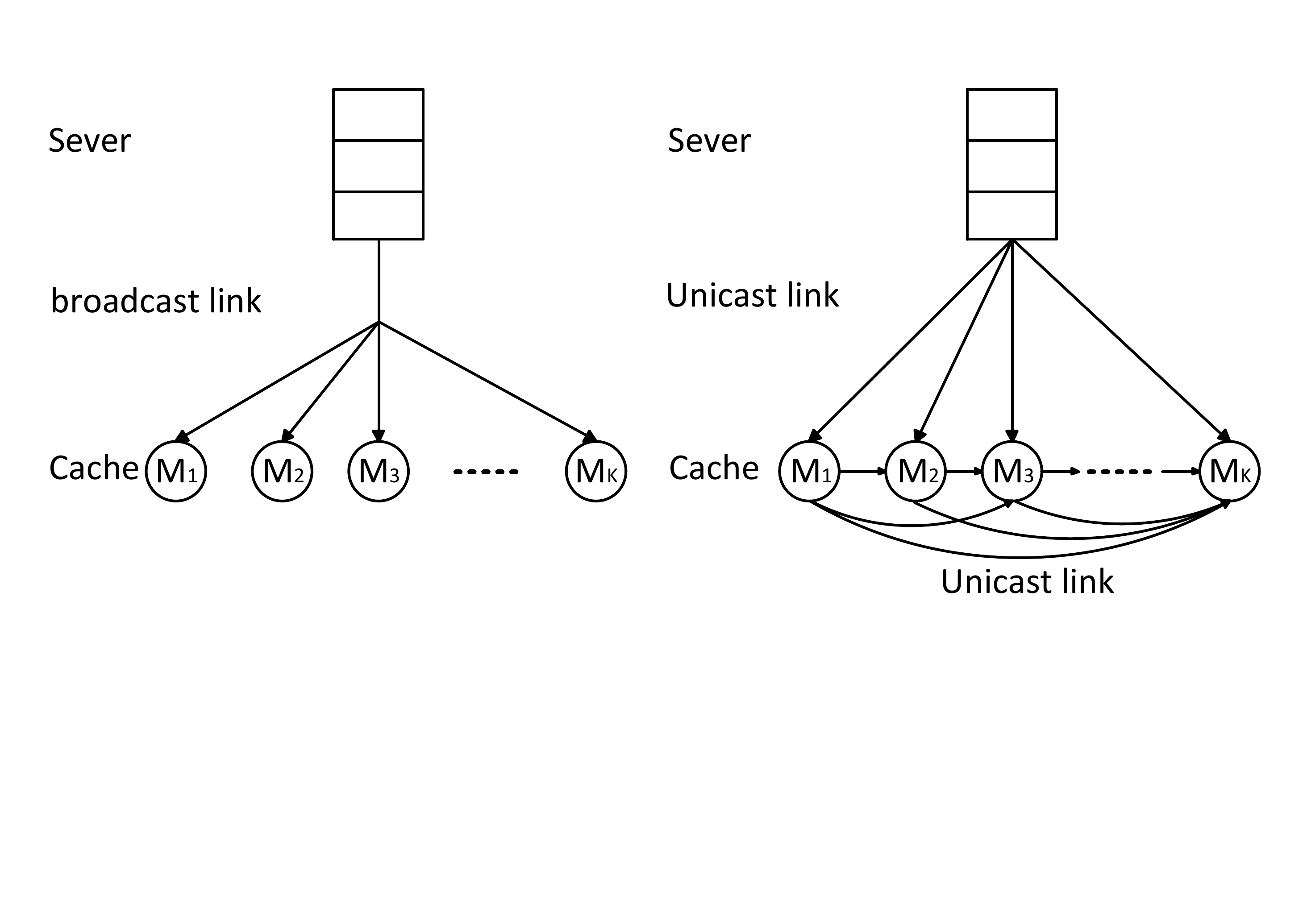}}
 \vspace{-2.5cm}
  \caption{Equivalent network diagram of coded caching.}
  \label{fig:equalmodel}
   \vspace{-0.3cm}
 \end{figure}

As shown in the right part of Fig.~\ref{fig:equalmodel}, the equivalent network diagram consists of a content server connecting to $K$ users through an error-free unicasting link. Each user can utilize the contents that being stored in the users having smaller cache size, and communicate with them by an unicasting link. The caching procedure is same as scheme $\mathfrak{F}_{o}$ and the delivery phase uses the hierarchical unicasting transmission. When user $k$ sends its request $d_{k}$, it first uses its local cache to reconstruct part of its requested content $W_{d_{k}}$, then it successively uses user $(k-1)$ to $1$'s local cache to reconstruct $W_{d_{k}}$, finally the server unicasts the part of content $W_{d_{k}}$ that are not stored in user $1$ to $k$'s local cache to user $k$. Thus, the traffic volume incurred in the unicast link between server and user $k$ is
\begin{equation*}
	\prod\limits_{i=1}^{K}\left(1-\frac{M_{i}}{N}\right).
\end{equation*}
Summming the traffic volume produced by all users, we can get the traffic formula that is identical to (\ref{eq:trafficformula2}).

From this equivalent network diagram, we can see more clearly how the coded multicast plays a role in ``connecting'' user’s local cache. For the traditional uncoded cache, each user only uses its local cache, while for the scheme Fo, each user can access other user’s cache that has a smaller cache size.

\begin{corollary}(Traffic volume under uniform cache set)\label{coro:uniformtraffic}
When the size of each user's local cache all equals to $M$, and the corresponding cache set is $\mathcal{M}_{u}=\{M,\ldots,M\}$. The number of users is $K$, we have,
\begin{equation*}
R_{\mathfrak{F}_{o}}(\mathcal{M}_{u},N,K)=K\left(1-\frac{M}{N}\right)\frac{N}{KM}\left[1-\left(1-\frac{M}{N}\right)^K\right],
\end{equation*}
\end{corollary}

This result has been exhibited in previous work~\cite{codedcaching2}, which can be regarded as the special case in our regime that the users' cache sizes are homogeneous.

\begin{corollary}(Traffic volume under singularity cache size)\label{coro:singletraffic}
When the the cache set is $\mathcal{M}_{s}=\{M,M,\ldots, (1+\alpha)M\}, \alpha>0$, and the number of users is $K$, we have,
\begin{equation*}
R_{\mathfrak{F}_{o}}(\mathcal{M}_{s},N,K)=R(\mathcal{M}_{u},N,K)-\alpha\frac{M}{N}\left(1-\frac{M}{N}\right)^{K-1}.
\end{equation*}
\end{corollary}

Corollary~\ref{coro:uniformtraffic} shows that, the increase of aggregate cache size $KM$ will reduce the traffic volume in an inverse proportional manner, as show in the term $N/KM$. However, Corollary~\ref{coro:singletraffic} shows that the increase of aggregate cache size $(K+\alpha)M$, i.e., the increase of parameter $\alpha$, will reduce the traffic in a linear manner. Thus, the effect of the heterogeneous cache set on the traffic volume has a different form compared with the homogeneous cases.

\begin{corollary}(Relation between two kinds of traffic-volume)\label{coro:uniformizationtraffic}
For caching problem with $N$ contents, $K$ users with ordered cache set $\mathcal{M}=\{M_1,M_2,\ldots, M_K\}$, we have,
\begin{equation}
R_{\mathfrak{F}_{o}}(\mathcal{M},N,K)\geq R_{\mathfrak{F}_{o}}(\mathcal{M}^{I},N,K),
\end{equation}
equal if only if $M_{k}=\mathbb{E}_{k}[M_{k}],\forall k\in \mathcal{K}$.
\end{corollary}
\begin{proof}
	Here is a proof sketch of Corollary~\ref{coro:uniformizationtraffic}. We regard this problem as the conditional minimum of $R_{\mathfrak{F}_{o}}(\mathcal{M},N,K)$ with respect to the summation of $M_{i}$ is constant. We first prove the following operation can strictly reduce the value of function: $M_{i}'=(M_{i}+M_{i+1})/2$ and $M_{i+1}'=(M_{i}+M_{i+1})/2$. Then, we use the contradictory induction to prove above function attains the minimum when each $M_{i}$ is equal. See~\cite{techreport} for more details.
\end{proof}

The Corollary~\ref{coro:uniformizationtraffic} is easy to understand for the scheme $\mathfrak{F}_{o}$, since the there is no miss coding phenomenon when the size of each user's local cache is identical, the traffic volume will be minimized.

\subsection{Order Optimality Analysis}

Based on above close-form expression of the traffic volume-memory trade-off and the cut-set bound, we have
\begin{theorem}(Order optimality of scheme $\mathfrak{F}_{o}$)\label{thm:order}
For caching problem with $N$ contents, $K(K\leq N)$ users with heterogenous cache set $\mathcal{M}=\{M_1,M_1,\ldots, M_K\}$, we have,
\begin{equation}
\frac{R_{\mathfrak{F}_{o}}(\mathcal{M},N,K)}{R^{*}(\mathcal{M},N,K)} \leq 12.
\end{equation}
\end{theorem}
\begin{proof}
The proof of Theorem~\ref{thm:order} is extremely complicated and we provide a sketch of it. We divide this proof into following 15 cases based on the relation among $N, p$ ($p$ is defined in the first insight of lower bound) and $12$, and prove that, in each case, such gap is less than 12. See~\cite{techreport} for more details. 
\end{proof}

It can be seen that the scheme $\mathfrak{F}_{o}$ exhibits the constant gap 12 to the cut-set bound. In fact, we will show in the numerical analysis that this gap is less than 4. As analyzed before, since the miss of coding opportunities in this regime will lead to the increase of the traffic volume, the only reason for such constant gap is that the increase of cut-set bound, i.e., a more increasing speed due to the heterogenous cache set. Here we discuss the main reason behind this trend. The cut-set bound argument presents an essence that the broadcasting signal of caching network plays a role in combining users' caches to joint reconstruct their requested contents, and for a user group $U$ of $s$ users, the minimum size of broadcasting signal will always be bounded by those $s$ minimum caches. If the cache sizes are distributed in a nonuniform manner, the gain coming from such combining information retrieval will be weaken, namely, the broadcasting signal is limited in this regime. Thus, although coding scheme $\mathfrak{F}_{o}$ will miss the coding opportunities, it still guarantee the order optimality.

Moreover, this constant gap is related to the deviation of the cache set. Consider two extreme cases: in the first case, the cache set has the smallest deviation that the all users' cache size is identical, then the gap has been proved to 12; in the second case, the cache set has the largest deviation that the cache size of half users is $0$ while anther half is $N$, then the traffic volume and the corresponding information-theoretical lower bound is all $K/2$, the constant gap reduced to $1$. This result implies that the lower bound might have a more increasing speed than the coded caching scheme with the deviation of the cache set increasing and the coded mulitcasting efficiency in such network will gradually degenerated to the uncoded manner. Then it is natural to ask how such gap relates to the deviation of the cache sizes. In the next Section, we will use the concept of probabilistic cache set to quantitatively investigate the relationship between this constant gap and the characteristics of the cache set.

\section{Theoretical Analysis \label{sec5}}

In this section, we assume the cache size of each user is a random variable and the corresponding cache set is referred to probabilistic cache set. We derive the expected gap between the traffic volume produced by scheme $\mathfrak{F}_{o}$ and the cut-set bound under the specific cache size distribution. Besides that, we preliminarily investigate the grouping effect of the coded caching scheme under heterogeneous cache size.

\subsection{Order Optimality under Probabilistic Cache Set}

Let $\mathbf{M_{K}}=\{\mathbf{M_{1,K}},\mathbf{M_{2,K}},\ldots,\mathbf{M_{K,K}}\}$ denote the cache set, where $\mathbf{M_{K}}$ is a set of order statistics that comes from a common parental distribution $\mathbf{F}$. Here, the order statistics represent the a series of independent random variables such that $\mathbf{M_{1,K}}\leq \mathbf{M_{2,K}}\leq \cdots \leq \mathbf{M_{K,K}}$. Further detail can be seen in~\cite{orderstatistics}.

Correspondingly, the order optimality is defined by
\begin{definition}(Expected order optimal)
The scheme $\mathfrak{F}$ is order optimal if only if
\begin{equation}\label{eq:deforderopt}
\mathbb{E}\left[\frac{R_{\mathfrak{F}}(\mathbf{M_{K}},N,K)}{R^{*}(\mathbf{M_{K}},N,K)}\right]  \leq  C(\mathbf{F}).
\end{equation}
\end{definition}

We can see that this constant is function of the parental distribution $\mathbf{F}$. In the following discussion, we will show that this constant is a function of $1$th moment and $2$th moment of $\mathbf{F}$ when $K=2$.

\begin{theorem}(Order optimality for normal distributed cache size)\label{thm:avorder}
For the caching problem, $N(N\geq2)$ contents, $K=2$ users each with cache size $\mathbf{M_{2}}=\{\mathbf{M_{1,2}},\mathbf{M_{2,2}}\}$, and comes from a normal distribution $\mathcal{N}(\mu,\sigma^2)$, then
\begin{equation}\label{eq:avorder}
\mathbb{E}\left[\frac{R_{\mathfrak{F}_{o}}(\mathbf{M_{2}},N,2)}{R^{c}(\mathbf{M_{2}},N,2)}\right]\leq 2-\frac{\sqrt{\pi}\mu+\sigma}{\sqrt{\pi}N}.
\end{equation}
\end{theorem}

From Theorem~\ref{thm:avorder}, we can clearly see how the constant gap related to the variance of the cache set, where the variance of the cache set plays a linear negative role in such constant. Similarly, for the uncoded caching scheme, such gap is less than
\begin{equation*}
2-\frac{2\sigma}{\sqrt{\pi}N},
\end{equation*}
which has a faster speed in the uncoded cache, seen $2\sigma$ rather than the $\sigma$ of coded version. Consider the traffic volume produce by the uncoded caching scheme $\mathfrak{F}_{u}$ is $\mathbb{R}_{\mathfrak{F}_{u}}=2-\frac{2\mu}{N}$ and irrelevant to the variance of cache set. Thus the descending gap of uncoded version mainly come from the increment of the cut-set bound, which further implies the deviation of cache set will strengthen the inherent degenerated performance of wireless coded mulitcasting. Remark that the variance $\sigma$ cannot be infinity since the range of $\mathbf{M_{K}}$ is $[0,N]$.

\subsection{Grouping Effect of Heterogenous Cache}

When the users' cache sizes are extremely different, the coded caching scheme will approximate to the GCD. This is illustrated in the following example.

\textbf{Example 5} (Group coded delivery) Suppose a system is distributing $N$ contents to $K$ users (here $K$ is an even number) with the cache set in first scenario is $\mathcal{M}_{d}=\{\mathcal{M}_{1},\mathcal{M}_{2}\}$, where $\mathcal{M}_{1}=\{M,\ldots,M\}$, $\mathcal{M}_{2}=\{\alpha M \ldots \alpha M\}, \alpha>0$, and $|\mathcal{M}_{1}|=|\mathcal{M}_{2}|=\frac{K}{2}$. Under the scheme $\mathfrak{F}_{o}$, we have
\begin{align*}\label{eq:dgrouptraffic}
R_{\mathfrak{F}_{o}}(\mathcal{M}_{1},N,\frac{K}{2})+Q_{\frac{K}{2}}\cdot R_{\mathfrak{F}_{o}}(\mathcal{M}_{2},N,\frac{K}{2}),
\end{align*}
where $Q_{\frac{K}{2}}=\left(1-M/N\right)^{\frac{K}{2}}$. Since $\alpha M\leq N$, we have $Q_{\frac{K}{2}}\geq \left(1-1/\alpha\right)^{\frac{K}{2}}$. Then the cache set $\mathcal{M}_{d}$ has a large deviation corresponds to the large value of parameter $\alpha$ and that $Q_{\frac{K}{2}}$ approximates to $1$. Thus, the traffic formula $R_{\mathfrak{F}_{o}}(\mathcal{M}_{d},N,K)$ will approximate to the traffic produced by two groups.

As for the general case that has $L$ groups, we use the concept of probabilistic cache set to develop an upper bound of the increment of traffic volume due to the GCD and show that this upper bound is related to the deviation of the cache set, i.e., the larger deviation, the smaller gap between coded cache and GCD.

\begin{theorem}(Expected upper bound of GCD)\label{thm:avgrouptraffic}
For caching problem with $N$ contents, $K(K<N)$ users with probabilistic cache set $\mathbf{M_{K}}$, we divide all users into $L$ groups $\mathcal{K}=\{\mathbf{K_{1}},\ldots,\mathbf{K_{L}}\}$, $\mathbf{M_{K}}=\{\mathbf{M_{K}^{1}},\ldots,\mathbf{M_{K}^{L}}\}$, then we have,
\begin{equation}\label{eq:upperbound}
\mathbb{E}\left[\frac{\sum_{i=1}^{L}R_{\mathfrak{F}_{o}}(\mathbf{M_{K}^{i}},N,|\mathbf{K_{i}}|)}{R_{\mathfrak{F}_{o}}(\mathbf{M_{K}},N,|\mathbf{K}|)}\right] \leq \frac{1}{1+e^{-\frac{K\mu}{N}}}L^{\frac{\mu}{\mu+\sigma}}.
\end{equation}
\end{theorem}

It can be seen that this bound is dependent on the expectation and variance of distribution $\mathbf{F}$ and shows a trend that the increment of the traffic volume will be reduced when the variance increases. This result means that the heterogenous cache set will diminish the effect of the GCD. The reason that we get this kind of form of this upper bound can be seen in the full version of this paper~\cite{techreport}. The Fig.~\ref{fig:numericalbound} partly shows the tightness and effectiveness of this bound.

\begin{figure}[htb]
\vspace{-0.3cm}
 \centerline{ \includegraphics[angle=0,width=0.49\textwidth]{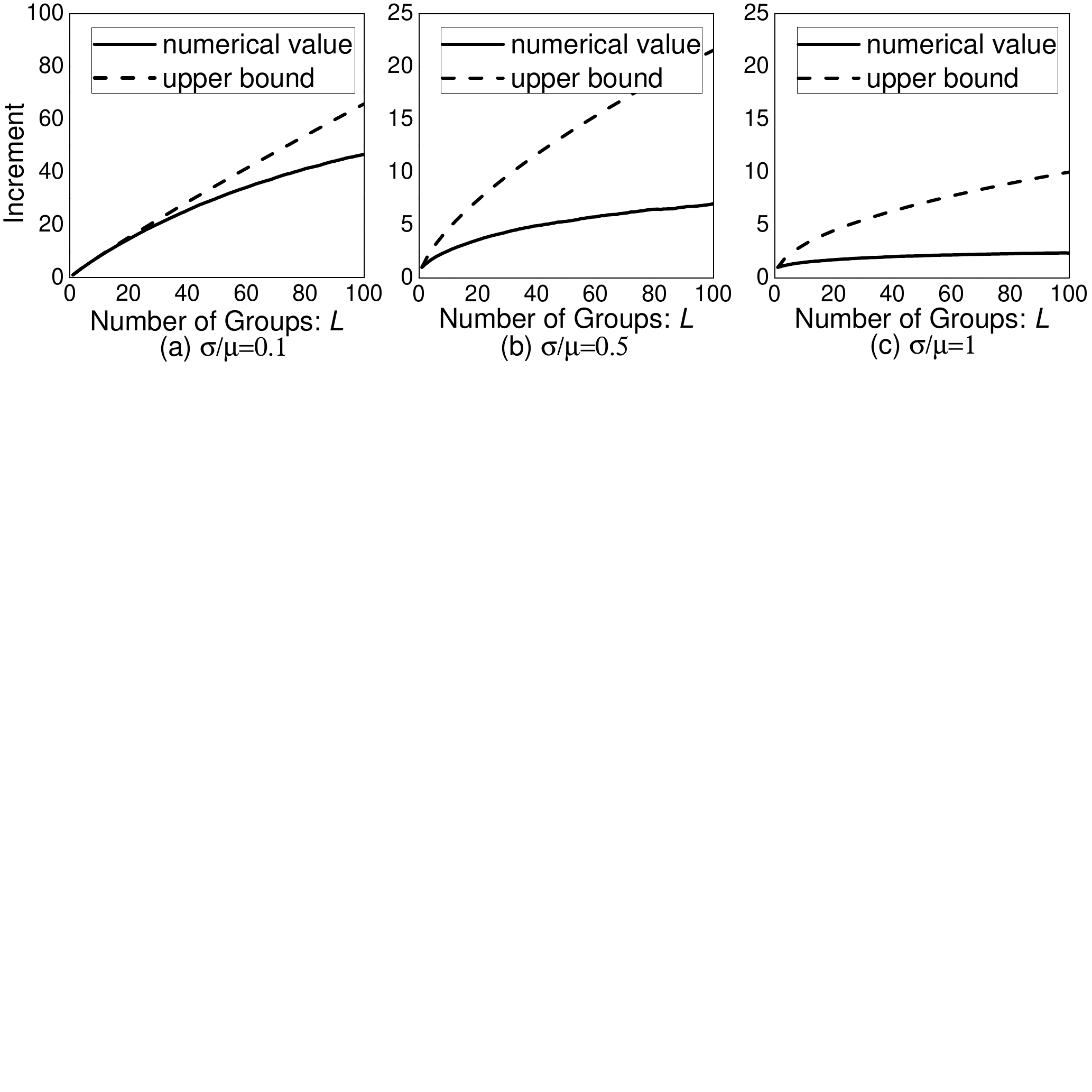}}
 \vspace{-6.2cm}
  \caption{The comparison between approximate upper bound (\ref{eq:upperbound}) and real value of increment of traffic volume under four kinds of deviation of cache set. The system parameters $N=500$, $K=300$ and $\mu=200$.}
  \label{fig:numericalbound}
   \vspace{-0.3cm}
 \end{figure}

As show in Fig.~\ref{fig:numericalbound}, when the variance is extremely small, our approximate bound is tight, especially when the number of groups is small, while when the variance is large, this approximate bound will be a little loose. Besides that, another important observation is that he increment of traffic volume is factually convergent to a constant. In particular, when the variance $\sigma=\mu$, this increment is bounded by only a constant factor 2.

Moreover, the grouping effect of the heterogenous cache size provides us with an insight that, in the practical scenario, the coded caching scheme should be operated in the grouping manner. As shown in Fig.~\ref{fig:numericalbound}(b), the variance $\sigma=0.5\mu$ to denote the most practical scenario. If we divide the 300 users into 15 groups and each group has 20 users, the complexity will be reduced from $2^300\approx 10^{90}$ to $15\cdot 2^{20}\approx 10^{7}$, while the increment of traffic volume is just two times.

\section{Numerical Results \label{sec6}}

In the previous section, all analysis is in the scenario that $K\leq N$. In this section, we present the numerical results when $K>N$. Then, we investigate the impact of the system parameters on the delivery-phase traffic volume. Moreover, we systematically investigate the performance of GCD.

\subsection{The Effect of Cache Set}

We have proven that when $K=2$, the constant gap is related to the variance of the cache set. For large $K$, we use the numerical analysis to show how the constant gap scales as the deviation of the cache set. To guarantee the deviation of the cache set, the average cache size cannot be too large. In fact, this constraint is reasonable, since the number of contents is mostly much larger than the cache size. Thus, we set the maximum average cache size $\mu=0.3N$, and the cache set comes from a normal parental distribution.

\begin{figure}[htb]
\vspace{-0.2cm}
 \centerline{ \includegraphics[angle=0,width=0.475\textwidth]{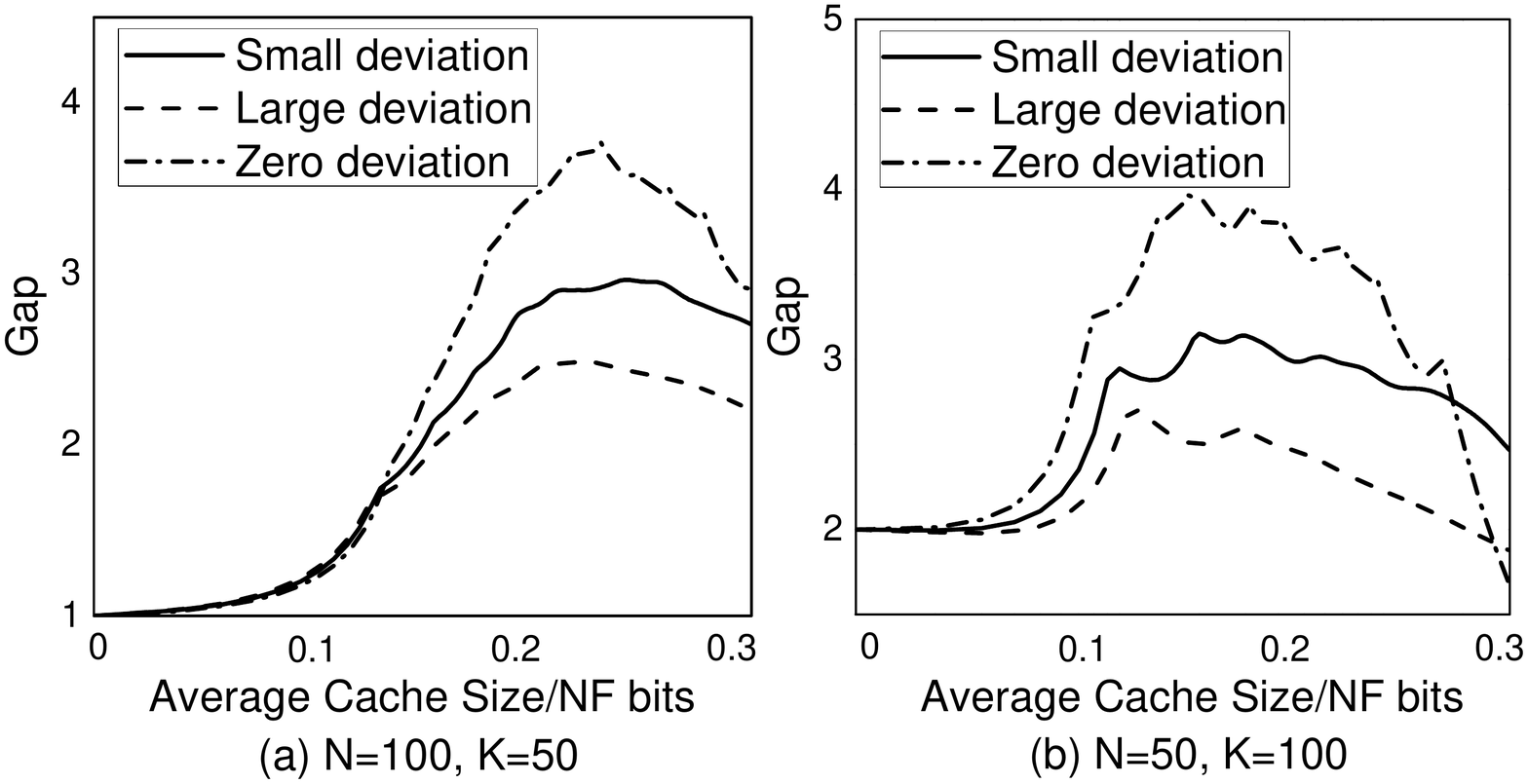}}
 \vspace{-2.3cm}
  \caption{The impact of average cache size on the gap between $R_{\mathfrak{F}_{o}}(\mathcal{M},N,K)$ and $R^{c}(\mathcal{M},N,K)$. Note that the scale in the (a) is a semi-logarithmic coordinate system.}
  \label{fig:simucacheset1}
   \vspace{-0.3cm}
 \end{figure}

The gap in Fig.~\ref{fig:simucacheset1} and Fig.~\ref{fig:simucacheset2} refers to the constant gap between traffic volume produced by scheme $\mathfrak{F}_{o}$ and cut-set bound. Fig.~\ref{fig:simucacheset1} plots the gap versus the average cache size $\mu$ under fixed variance of cache set. The variance under small and large deviation are denoted by $\sigma=0.1\mu$ and $\sigma=0.7\mu$, separately. The zero deviation refers to the traditional case that the cache size is uniform and we regard it as a baseline. In the Fig.~\ref{fig:simucacheset1}(a), the number of contents $N=100$ is larger than the number of users $K=50$. The gap under large deviation is strictly less than that under small deviation, and less than the baseline.  While in the Fig.~\ref{fig:simucacheset1}(b) that $N=50$, $K=100$, this relation shows piecewise characteristics: when the average cache size is small ($\mu<0.25N$), it has the same manner as $N>K$, when the average cache size is large ($\mu>0.25N$), the gap under homogenous cache set decreases dramatically and less than the other two cases.

\begin{figure}[htb]
\vspace{-0.3cm}
 \centerline{ \includegraphics[angle=0,width=0.475\textwidth]{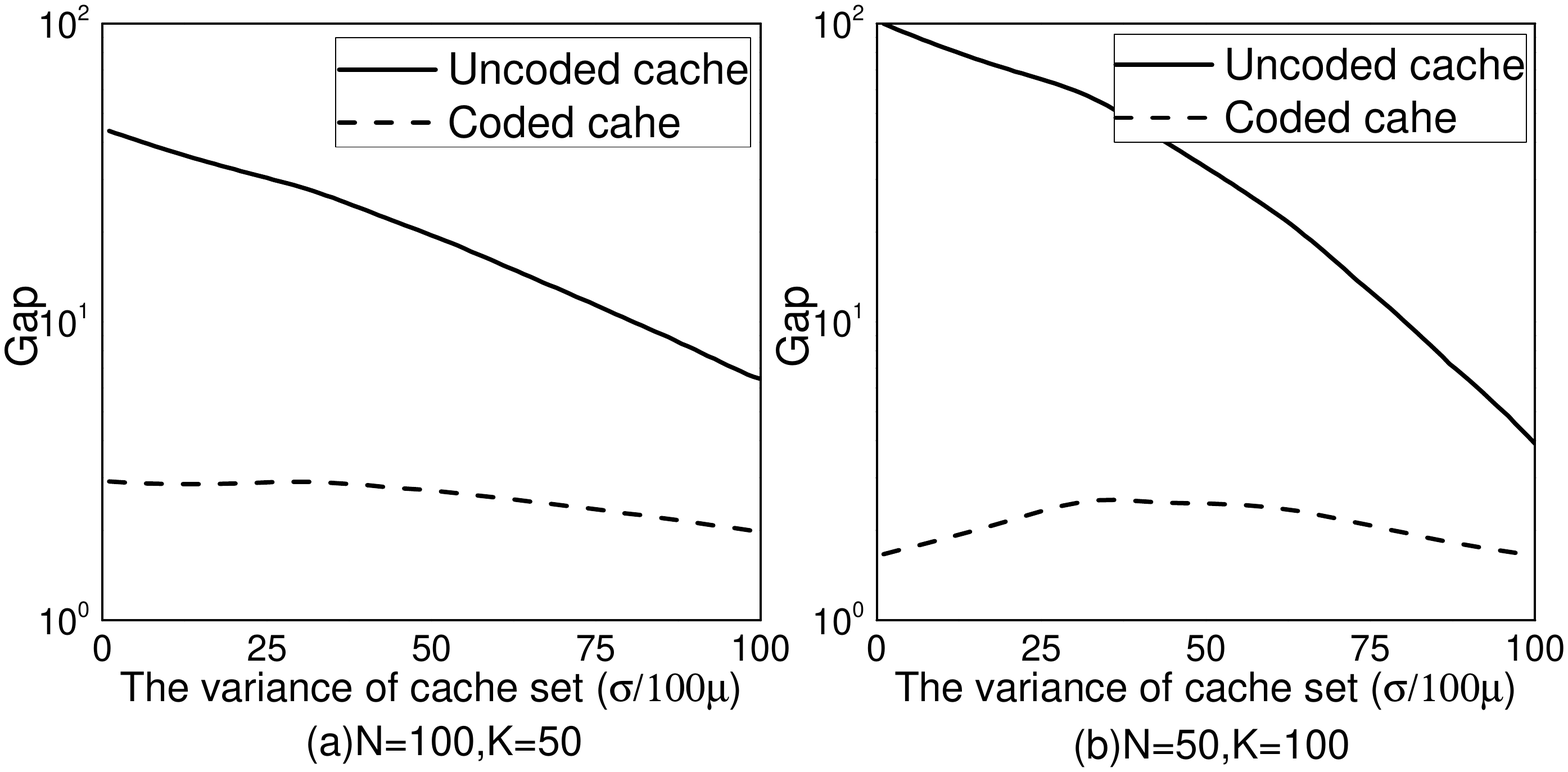}}
 \vspace{-2.3cm}
  \caption{The impact of deviation of cache set on the gap.}
   \label{fig:simucacheset2}
   \vspace{-0.5cm}
 \end{figure}

Fig.~\ref{fig:simucacheset2} plots the gap versus the variance of the cache set under fixed average cache size $\mu=30$. It can be seen that, the performance gain of coded cache compared with the uncoded cache is $15\times$ under small deviation ($\sigma=0.05\mu$), while only approximately $3\times$ under large deviation ($\sigma=0.6\mu$). For the case that $N>K$, seen in Fig~\ref{fig:simucacheset2}(a), both gap-variance curves of coded and uncoded cache have a declining trend and uncoded version performs faster. While for the case that $K<N$, seen in Fig~\ref{fig:simucacheset2}(b), the gap-variance curve of coded cache shows an unique single valley manner: under small cache sizes, it first increases, then it decreases linearly when the variance is large.

\subsection{The Impact of number of contents and users: $N$ and $K$}

Then we present how the gap scales as the system parameters such as $N$ and $K$. Assume there are two relationships between $K$ and $N$: $N=\Theta(K)$ and $N=\omega(K)$. The first case refers to the number of users and the number of contents have the same order. The second case refers to the number of content is extremely larger that the number of users. The characteristic of cache set is $\mu=30,\sigma=0.3\mu$.

\begin{figure}[htb]
\vspace{-0.0cm}
 \centerline{ \includegraphics[angle=0,width=0.475\textwidth]{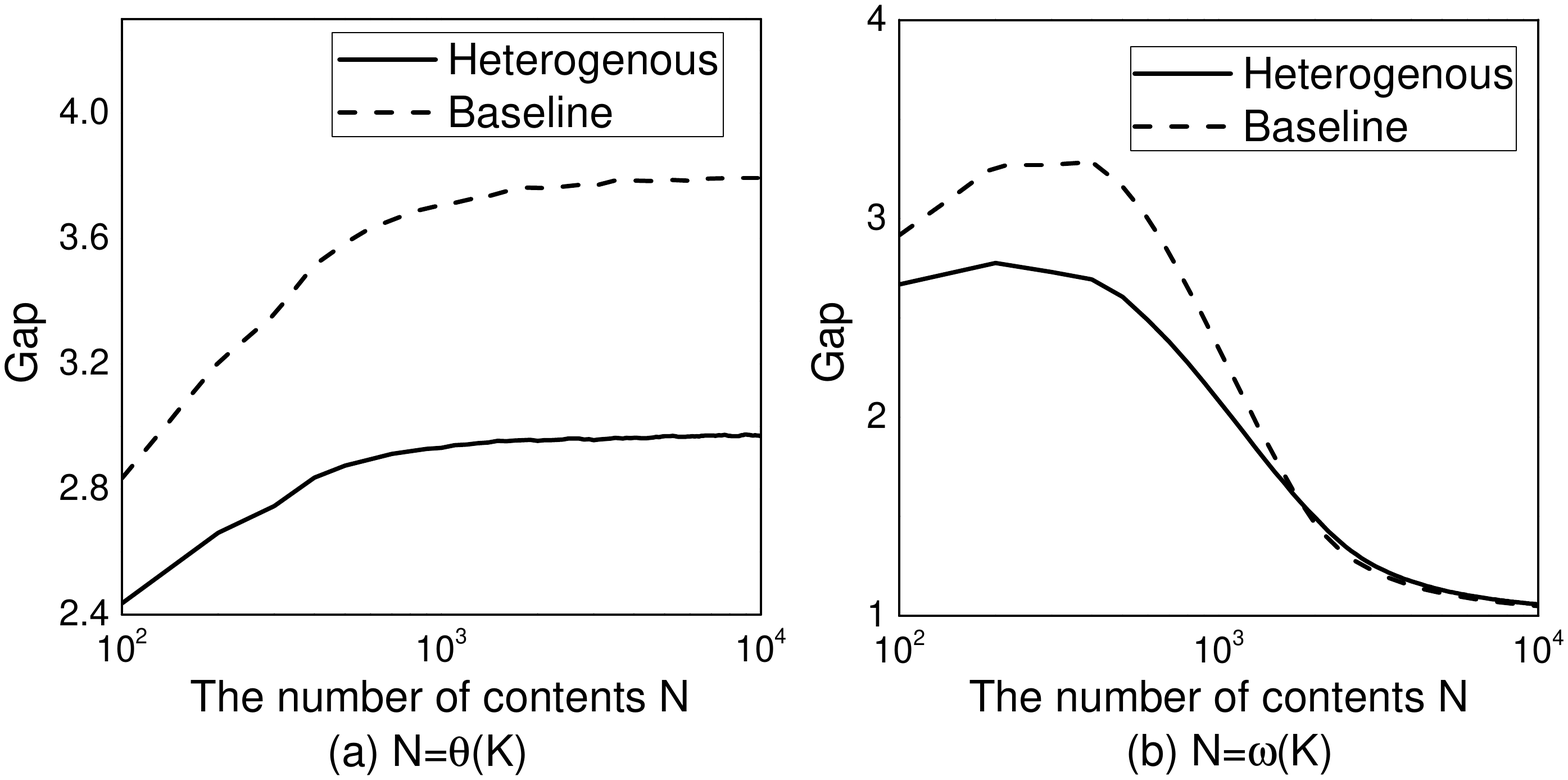}}
 \vspace{-2.3cm}
  \caption{The impact of system parameters $N$ and $K$ on the gap between $R_{\mathfrak{F}_{o}}(\mathcal{M},N,K)$ and $R^{c}(\mathcal{M},N,K)$. The average cache size is $\mu=30.$}
  \label{fig:simuNK}
   \vspace{-0.3cm}
 \end{figure}

In Fig.~\ref{fig:simuNK}(a), we can see that the gap gradually increases to a constant when $N$ is sufficiently large, while in Fig.~\ref{fig:simuNK}(b), the gap gradually decreases to $1$. Based on the asymptotic analysis, we can get easily get the reason for $N=\omega(K)$. Since $N\rightarrow\infty$, each term in (\ref{eq:trafficformula2}) will approximate to $1$ and $R_{\mathfrak{F}_{o}}(\mathcal{M},N,K)$ will approximate to $K$. In the similar manner, the cut-set bound $R^{c}(\mathcal{M},N,K)$ will also approximate to $K$. Thus the gap will approximate to 1. The reason for $N=\Theta(K)$ is based on an extremely complicated series analysis and can be seen in the technical report~\cite{techreport}.

\section{Conclusion and Open Problems \label{sec7}}

In this paper, we have investigated the fundamental limits of coded cache under the heterogenous cache set. Through deriving the fundamental bound in this regime, we have pointed out that coded caching scheme with zero-bits-padding can still guarantee the order optimality. Moreover, using the concept of probabilistic cache set, we have proven such gap is closely related to the deviation of cache set. And above results show the inherent degenerated performance of the bottleneck network under the heterogenous cache sizes. Besides that, we also show the grouping effect of the coded caching scheme and this result provide an insight with us that the coded caching scheme under the heterogenous setting should be operated in the grouping manner.

There are several important open problems induced by the heterogenous cache framework:
\begin{itemize}
	\item In this work, we point out the necessity of GCD under heterogenous cache. What is the optimal grouping plan if the number of group is given?
	\item The content popularity and heterogenous cache size will both cause the miss of coding opportunities. How is the performance when we consider them simultaneous?
	\item This work focuses on the worst-case analysis, how is the performance if we study this problem in the aspect of the average performance?
\end{itemize}

\bibliographystyle{ieeetr}

\end{document}